\documentclass{article}

\usepackage{epsfig}
\usepackage{amsfonts}
\usepackage{amssymb}
\usepackage{amstext}
\usepackage{amsmath}
\usepackage{algorithm}
\usepackage{verbatim}
\usepackage{amsthm}
\usepackage[normalem]{ulem}


\numberwithin{equation}{section}
\numberwithin{table}{section}

\newtheorem{fact}{Fact}[section]
\newtheorem{lemma}[fact]{Lemma}
\newtheorem{theorem}[fact]{Theorem}

\newcommand{\polylog}{\operatorname{polylog}}

\newcommand{\junk}[1]{}

\newcommand{\vol}{{\mbox{vol}}}

\title{Finding Small Sparse Cuts Locally by Random Walk\footnote{This is independent from the work~\cite{gharan-trevisan} which obtained similar results.}}

\author{ 
Tsz Chiu Kwok, Lap Chi Lau \vspace*{2mm}\\
The Chinese University of Hong Kong}

\date{}

\begin{document}
\maketitle

\begin{abstract}
We study the problem of finding a small sparse cut in an undirected graph.
Given an undirected graph $G=(V,E)$ and a parameter $k \leq |E|$,
the small sparsest cut problem is to find a set $S \subseteq V$ with minimum conductance among all sets with volume at most $k$.
Using ideas developed in local graph partitioning algorithms, 
we obtain the following bicriteria approximation algorithms for the small sparsest cut problem:
\begin{itemize}
\item If there is a set $U \subseteq V$ with conductance $\phi$ and $\vol(U) \leq k$, then there is a polynomial time algorithm to find a set $S$ with conductance $O(\sqrt{\phi/\epsilon})$ and $\vol(S) \leq k^{1+\epsilon}$ for any $\epsilon>1/k$.
\item If there is a set $U \subseteq V$ with conductance $\phi$ and $\vol(U) \leq k$, then there is a polynomial time algorithm to find a set $S$ with conductance $O(\sqrt{\phi \log k / \epsilon})$ and $\vol(S) \leq (1+\epsilon)k$ for any $\epsilon>2\ln k/k$.
\end{itemize}
These algorithms can be implemented locally using truncated random walk,
with running time almost linear to the output size.
This provides a local graph partitioning algorithm with a better conductance guarantee when $k$ is sublinear.

\end{abstract}


\section{Introduction}

For an undirected graph $G=(V,E)$,
the conductance of a set $S \subseteq V$ is defined as $\phi(S) = |\delta(S)|/\vol(S)$, where $\delta(S)$ is the set of edges with one endpoint in $S$ and another endpoint in $V-S$, and $\vol(S) = \sum_{v \in S} d(v)$ where $d(v)$ is the degree of $v$ in $G$.
Let $n=|V|$ and $m=|E|$.
The conductance of $G$ is defined as $\phi(G)=\min_{S:\vol(S) \leq m} \phi(S)$.
The conductance of a graph is an important parameter that is closely related to the expansion of a graph and the mixing time of a random walk~\cite{hoory-linial-wigderson}.
Finding a set of small conductance, called a sparse cut, is a well-studied algorithmic problem that has applications in different areas.
Several approximation algorithms are known for the sparsest cut problem.
The spectral partitioning algorithm by Cheeger's inequality~\cite{cheeger,alon-milman} finds a set of of conductance $\sqrt{\phi(G)}$ with volume at most $m$.
The linear programming rounding algorithm by Leighton and Rao~\cite{leighton-rao} finds a set of conductance $O(\phi(G)\log(n))$ with volume at most $m$.
The semidefinite programming rounding algorithm by Arora, Rao and Vazirani~\cite{arora-rao-vazirani} finds a set of conductance $O(\phi(G) \sqrt{\log(n)})$ with volume at most $m$.

Recently there has been much interest in studying the small sparsest cut problem, to determine $\phi_k(G) = \min_{S:\vol(S) \leq k} \phi(S)$ for a given $k$, and to find a set of smallest conductance among all sets of volume at most $k$.
This is also known as the expansion profile of the graph~\cite{lovasz-kannan,raghavendra-steurer-tetali}.
There are two main motivations for this problem.
One is the small set expansion conjecture~\cite{raghavendra-steurer}, which states that for every constant $\epsilon>0$ there exists a constant $\delta>0$ such that it is NP-hard to distinguish whether $\phi_{\delta m}(G) \leq \epsilon$ or $\phi_{\delta m}(G) \geq 1-\epsilon$.
This conjecture is shown to be closely related to the unique games conjecture~\cite{raghavendra-steurer}, and so it is of interest to understand what algorithmic techniques can be used to estimate $\phi_k(G)$.
There are bicriteria approximation algorithms for this problem using semidefinite programming relaxations:
Raghavendra, Steurer and Tetali~\cite{raghavendra-steurer-tetali} obtained an algorithm that finds a set $S$ with $\vol(S) \leq O(k)$ and $\phi(S) \leq O(\sqrt{\phi_k(G) \log(m/k)})$, and
Bansal et.al.~\cite{bansal+} obtained an algorithm that finds a set $S$ with $\vol(S)\leq (1+\epsilon)k$ and $\phi(S) \leq O(f(\epsilon)~ \phi_k(G) \sqrt{\log n \log(m/k)})$ for any $\epsilon>0$ where $f(\epsilon)$ is a function depends only on $\epsilon$.

Another motivation is the design of local graph partitioning algorithms in massive graphs. 
In some situations, we have a massive graph $G=(V,E)$ and a vertex $v \in V$, 
and we would like to identify a small set $S$ with small conductance that contains $v$ (if it exists).
The graph may be too big that it is not feasible to read the whole graph and run some nontrivial approximation algorithms.
So it would be desirable to have a local algorithm that only explores a small part of the graph, and outputs a set $S$ with small conductance that contains $v$, and the running time of the algorithm depends only on $\vol(S)$ and $\polylog(n)$.
All local graph partitioning algorithms are based on some random walk type processes.
The efficiency of the algorithm is measured by the work/volume ratio, which is defined as the ratio of the running time and the volume of the output set.
Spielman and Teng~\cite{spielman-teng} proposed the first local graph partitioning algorithm using truncated random walk, that returns a set $S'$ with $\phi(S') = O(\phi^{1/2}(S) \log^{3/2} n)$ if the initial vertex is a random vertex in $S$, and the work/volume ratio of the algorithm is $O(\phi^{-2}(S) \polylog(n))$.
Anderson, Chung, Lang~\cite{anderson-chung-lang} used local pagerank vectors to find a set $S'$ with $\phi(S') = O(\sqrt{\phi(S) \log k})$ and work/volume ratio $O(\phi^{-1}(S) \polylog(n))$, if the initial vertex is a random vertex in a set $S$ with $\vol(S)=k$.
Anderson and Peres~\cite{anderson-peres} used the volume-biased evolving set process to obtain a local graph partitioning algorithm with work/volume ratio $O(\phi^{-1/2} \polylog(n))$ and a similar conductance guarantee as in \cite{anderson-chung-lang}.
Note that the running time of these algorithms would be sublinear if the volume of the output set is small, which is the case of interest in massive graphs.

\subsection{Main Results}

We show that the techniques developed in local graph partitioning algorithms~\cite{spielman-teng,chung} can be used to obtain bicriteria approximation algorithms for the small sparsest cut problem.
We obtain a tradeoff between the conductance guarantee and the volume of the output set.

\begin{theorem} \label{thm:main}
Given an undirected graph $G=(V,E)$ and a parameter $k$,
there is a polynomial time algorithm to do the following:
\begin{enumerate}
\item
Find a set $S$ with $\phi(S) = O(\sqrt{\phi_k(G)/\epsilon})$ and $\vol(S) \leq k^{1+\epsilon}$ for any $\epsilon>1/k$.
\item
Find a set $S$ with $\phi(S) = O(\sqrt{\phi_k(G) \log k/\epsilon})$ and $\vol(S) \leq (1+\epsilon)k$ for any $\epsilon>2 \ln k/k$.
\end{enumerate}
\end{theorem}



For the small sparsest cut problem, when $k$ is sublinear ($k=O(m^c)$ for $c<1$), the performance guarantee of the bicriteria approximation algorithm in Theorem~\ref{thm:main}(2) is similar to that of Raghavendra, Steurer and Tetali~\cite{raghavendra-steurer-tetali}.
Also, when $k$ is sublinear, the conductance guarantee of Theorem~\ref{thm:main}(1) is independent of $n$, which matches the performance of spectral partitioning while having a bound on the volume of the output set.
These show that random walk algorithms can also be used to give nontrivial bicriteria approximations for the small sparsest cut problem.
Moreover the algorithms can be implemented locally by using the truncated random walk algorithm. 

\begin{theorem} \label{thm:local}
For an undirected graph $G=(V,E)$ and a set $U \subseteq V$,
given $\phi \geq \phi(U)$ and $k \geq \vol(U)$,
there exists an initial vertex such that the truncated random walk algorithm can find a set $S$ with $\phi(S) \leq O(\sqrt{\phi/\epsilon})$ and $\vol(S) \leq O(k^{1+\epsilon})$ for any $\epsilon>2/k$.
The work/volume ratio of the algorithm is $O(k^{\epsilon} \phi^{-2})$.
\end{theorem}

When $k$ is sublinear, the interesting case of local graph partitioning algorithms,
the conductance guarantee of Theorem~\ref{thm:local} matches that of spectral partitioning, improving on the conductance guarantees in previous local graph partitioning algorithms.
However, we note that our notion of a local graph partitioning algorithm is much weaker than previous work~\cite{spielman-teng,anderson-chung-lang,anderson-peres}, where they proved that a random initial vertex $u$ will work with a constant probability.
We only prove that there exists an initial vertex that will work, and unable to prove the high probability statement.

In Section 4 we discuss a connection to the small set expansion conjecture.

\subsection{Techniques}

The techniques are from the work of Spielman and Teng~\cite{spielman-teng} and Chung~\cite{chung}.
Our goal in Theorem~\ref{thm:main}(1) is equivalent to distinguish the following two cases:
(a) there is a set $S$ with $\vol(S) \leq k$ and $\phi(S) \leq \varphi$, or
(b) the conductance of every set of volume at most $ck$ is at least $\Omega(\sqrt{\varphi})$ for some $c > 1$.
As in~\cite{spielman-teng}, we use the method of Lov\'asz and Simonovits~\cite{lovasz-simonovits} that considers the total probability of the $k$ edges with largest probability after $t$ steps of random walk, call this number $C_t(k)$.
In case (a), we use the idea of Chung~\cite{chung} that uses the local eigenvector of $S$ of the Laplacian matrix to show that there exists an initial vertex such that $C_t(k) \geq (1-\frac{\varphi}{2})^t$.
In case (b), we use a result of Lov\'asz and Simonovits~\cite{lovasz-simonovits} to show that $C_t(k) \leq \frac{1}{c} + \sqrt{k}(1-M\varphi)^t$ for a large enough constant $M$, no matter what is the initial vertex of the random walk.
Hence, say when $c \geq k^{0.01}$, 
by setting $t=\Theta(\log k/\varphi)$,
we expect that $C_t(k)$ is significantly greater than $1/c$ in case (a) but at most $1/c$ plus a negligible term in case (b), and so we can distinguish the two cases.
To prove Theorem~\ref{thm:local}(1), we use the truncated random walk algorithm as in~\cite{spielman-teng} to give a bound on the work/volume ratio.
Theorem~\ref{thm:main}(2) is a corollary of Theorem~\ref{thm:main}(1).

\section{Finding Small Sparse Cuts} \label{s:main}




The organization of this section is as follows.
First we review some basics about random walk in undirected graphs.
Then we present our algorithm in Theorem~\ref{thm:main} and the proof outline, and then we present the analysis and complete the proof of Theorem~\ref{thm:main}.

\subsection{Random Walk}

In the following we assume $G = (V, E)$ is a simple unweighted undirected connected graph with $n = |V|$ vertices and $m = |E|$ edges.
Our algorithms are based on random walk.
Let $p_0$ be an initial probability distribution on vertices.
Let $A$ be the adjacency matrix of $G$, 
$D$ be the diagonal degree matrix of $G$,
and $W = \frac12(I + D^{-1}A)$ be the lazy random walk matrix.
The probability distribution after $t$ steps of lazy random walk is defined as $p_t = p_0W^t$.
(For convenience, we use $p_t$ to denote a row vector, while all other vectors by default are column vectors.)
For a subset $S \subseteq V$, we use $p_t(S)$ to denote $\sum_{u \in S} p_t(u)$.

To analyze the probability distribution after $t$ steps of lazy random walk, we use the method developed by Lov\'asz and Simonovits~\cite{lovasz-simonovits} as in other local graph partitioning algorithms~\cite{spielman-teng,anderson-chung-lang}.
We view the graph as directed by replacing each undirected edge with two directed edges with opposite directions.
Given a probability distribution $p$ on vertices, each directed edge $e = uv$ is assigned probability $q(e) = p(u)/d_u$.
Let $e_1, e_2, \ldots, e_{2m}$ be an ordering of the directed edges 
such that $q(e_1) \geq q(e_2) \geq \ldots \geq q(e_{2m})$.
The curve introduced by Lov\'asz and Simonovits $C : [0, 2m] \to [0, 1]$ is defined as follows:
for integral $x$, $C(x) = \sum_{i=1}^{x} q(e_i)$;
for fractional $x = \lfloor x \rfloor + r$, $C(x) = (1 - r)C(\lfloor x \rfloor) + rC(\lceil x \rceil)$.
Let $C_t$ be the curve when the underlying distribution is $p_t$.
Let $v_1, v_2, \ldots, v_n$ be an ordering of the vertices such that $p_t(v_1)/d(v_1) \ge p_t(v_2)/d(v_2) \ge \dots \ge p_t(v_n)/d(v_n)$.
Then $C_t(\sum_{i=1}^j d(v_i)) = \sum_{i = 1}^j p_t(v_i)$ for all $j \in [n]$.
We call the points $x_j = \sum_{i = 1}^j d(v_i)$ extreme points, and note that the curve is linear between two extreme points.
We also call the sets $S_{t, j} = \{v_1, \ldots, v_j\}$ for $1 \leq j \leq n$ the level sets at time $t$.

The curve $C_t$ is concave, and it approaches the straight line $x/(2m)$ when $p_t$ approaches the stationary distribution.
Lov\'asz and Simonovits~\cite{lovasz-simonovits} analyzed the convergence rate of this curve to the straight line based on the conductances of the level sets.

\begin{lemma}[Lov\'asz-Simonovits~\cite{lovasz-simonovits}]
\label{l:bc} 
Let $x = x_j \le m$ be an extreme point at time $t$ and $S = S_{t, j}$ be the corresponding level set.
If $\phi(S) \ge \varphi$, then $C_{t}(x) \le \frac12(C_{t-1}(x - \varphi x) + C_{t-1}(x + \varphi x))$.
\end{lemma}

\subsection{Algorithm}

Our algorithm is simple.
For each vertex $v$, we use it as the initial vertex of the random walk, and compute the probability distributions $p_t$ for $1 \leq t \leq O(n^2 \ln n)$.
Then we output the set of smallest conductance among all level sets $S_{t,j}$ (of all initial vertices) of volume at most $ck$, where in Theorem~\ref{thm:main}(1) we set $c=k^{\epsilon}$ and in Theorem~\ref{thm:main}(2) we set $c=1+\epsilon$.
Clearly this is a polynomial time algorithm.

To analyze the performance of the algorithm, we give upper and lower bound on the curve based on the conductances. 
On one hand, we use Lemma~\ref{l:bc} to prove that if all level sets of volume at most $ck$ are of conductance at least $\phi_1$, then the curve satisfies $C_t(x) \le f_t(x) := \frac{x}{ck} + \sqrt x(1 - \frac{\phi_1^2}8)^t$ for all $x \le k$.
Informally, this says that if $\phi_1$ is large, then $C_t(k)$ is at most $1/c$ plus a negligible term when $t$ is large enough.
This statement holds regardless of the initial vertex of the random walk.
On the other hand, if there exists a set $S$ of volume at most $k$ with conductance $\phi_2$, then we use the idea of Chung~\cite{chung} that uses the local eigenvector of $S$ of the Laplacian matrix to show that there exists an initial vertex for which $C_t(k) \geq (1-\frac{\phi_2}{2})^t$.
Informally, this says that if $\phi_2$ is small, then $C_t(k)$ is significantly larger than $1/c$ if $c$ is large.
Finally, by combining the upper and lower bound for $C_t(k)$ and choosing an appropriate $t$, we show that $\phi_1 \leq O(\sqrt{\phi_2})$ when $c=k^{\epsilon}$ and $\phi_1 \leq O(\sqrt{\phi_2 \ln k})$ when $c=1+\epsilon$.
Hence the algorithm can find a level set with the required conductance.

\subsection{Upper Bound}

We prove the upper bound using Lemma~\ref{l:bc}.
We note that the following statement is true for any initial probability distribution, in particular when $p_0 = \chi_v$ for any $v$.

\begin{theorem}
\label{t:ub} 
Suppose for all $t' \le t$ and $i \in [n]$, we have $\phi(S_{t', i}) \ge \phi_1$whenever $\vol(S_{t', i}) \le l \le m$.
Then the curve satisfies $C_t(x) \le f_t(x) := \frac xl + \sqrt x(1 - \frac{\phi_1^2}8)^t$ for all $x \le l$.
\end{theorem}
\begin{proof}
Let the extreme points $x_i$ satisfy $0 = x_0 \le x_1 \le x_2 \le \dots \le x_i \le l < x_{i + 1}$.
Note that $C_t$ is linear between extreme points and between $x_i$ and $l$, and $f_t$ is concave.
So we only need to show the inequality for extreme points and the point $l$.
At the point $x = l$, the inequality always hold as $f_t(l) \ge 1 \ge C_t(l)$ for any $t$.
Now we would prove by induction.
When $t = 0$ the inequality is trivial as $f_0(x) \ge 1 \ge C_0(x)$ for all $x \ge 1$.
When $t > 0$ and $x$ is an extreme point,
\begin{eqnarray*}
C_t(x) & \le & \frac12(C_{t - 1}(x - \phi_1 x) + C_{t - 1}(x + \phi_1 x)) \quad {\rm (by~Lemma~\ref{l:bc})}\\
& \le & \frac12(f_{t - 1}(x - \phi_1 x) + f_{t - 1}(x + \phi_1 x)) \quad {\rm (by~induction)}\\
& = & \frac{x}{l} + \frac12 \sqrt x(1 - \frac{\phi_1^2}8)^{t - 1}(\sqrt{1 - \phi_1} + \sqrt{1 + \phi_1})\\
& \le & \frac{x}{l} + \sqrt x(1 - \frac{\phi_1^2}8)^t,
\end{eqnarray*}
where the last inequality follows from Taylor expansions of $\sqrt{1-\phi_1}$ and $\sqrt{1+\phi_1}$.
\end{proof}

\subsection{Lower Bound}

The idea is to use the local eigenvector of $S$ of the normalized Laplacian matrix to show that there is an initial distribution such that $p_t(S) \geq (1-\frac{\phi_2}{2})^t$.

\begin{theorem}
\label{t:lb} 
Assume $S \subseteq V$ where $\vol(S) \le m$ and $\phi(S) \le \phi_2$.
Then there exists a vertex $v$ such that if $p_0 = \chi_v$, then $p_t(S) \ge (1 - \frac{\phi_2}{2})^t$.
\end{theorem}
\begin{proof}
Let $\mathcal L = D^{-\frac12}LD^{-\frac12}$ be the normalized Laplacian matrix, where $L=D-A$ is the Laplacian matrix of the graph.
For any matrix $M$ with rows and columns indexed by $V$, let $M_S$ be the $|S| \times |S|$ submatrix of $M$ with rows and columns indexed by the vertices in $S$.
Consider the smallest eigenvalue $\lambda_S$ of $\mathcal L_S$ and its corresponding eigenvector $v_S$.
Let $\chi_S$ be the characteristic vector of $S$.
We have
\[(D_S^{\frac12}\vec1)^T\mathcal L_S(D_S^{\frac12}\vec1) = \vec1^TL_S\vec1
= \sum_{e = uv \in E} (\chi_S(u) - \chi_S(v))^2 = |\delta(S)|.\] 
So, by the Courant-Fischer theorem,
\[\lambda_S \le \frac{(D_S^{\frac12}\vec1)^T\mathcal L_S(D_S^{\frac12}\vec1)}{\|D_S^{\frac12}\vec1\|_2^2} = \frac{|\delta(S)|}{\vol(S)}
\le \phi_2.\]
We assume without loss of generality that $S$ is a connected subgraph.
Then, by the Perron-Frobenius theorem, 
the eigenvector $v_S$ can be assumed to be positive,
and we can rescale $v_S$ such that $\|D_S^{\frac12}v_S\|_1 = 1$ is a probability distribution.
Let $p_{t, S}$ denote the restriction of $p_t$ on $S$.
We set the initial distribution $p_0$ such that $p_{0, S} = (D_S^{\frac12}v_S)^T$, and $p_{0,V-S}=0$.
We would show that $p_{t, S} \ge (1 - \frac{\lambda_S}{2})^t p_{0, S}$ by induction.
Clearly the statement is true when $t = 0$.
For $t > 0$, we have
\begin{eqnarray*}
p_{t, S} & \ge & p_{t - 1, S}W_S\\
& = & p_{t-1,S} \cdot \frac{(I_S+D_S^{-1}A_S)}{2}\\
& \ge & (1 - \frac{\lambda_S}{2})^{t - 1}v_S^TD_S^{\frac12} \frac{(I_S + D_S^{-1}A_S)}{2} \quad {\rm (by~induction)}\\
& = & (1 - \frac{\lambda_S}{2})^{t - 1}v_S^T(I - \frac{\mathcal L_S}{2})D_S^{\frac12}\\
& = & (1 - \frac{\lambda_S}{2})^tv_S^TD_S^{\frac12}\\
& = & (1 - \frac{\lambda_S}{2})^tp_{0, S}.
\end{eqnarray*}
Therefore,
\[p_t(S) = p_{t, S}(S) \ge (1 - \frac{\lambda_S}{2})^tp_{0, S}(S) \ge (1 - \frac{\phi_2}{2})^t.\]
Since random walk is linear and $v_S$ is a convex combination of $\chi_v$ where $v \in S$, there exists a vertex $v \in S$ such that if $p_0 = \chi_v$, then $p_t(S) \ge (1 - \frac{\phi_2}{2})^t$.
\end{proof}

\subsection{Proof of Theorem~\ref{thm:main}}

We combine the upper bound and the lower bound to prove Theorem~\ref{thm:main}.
We note that Theorem~\ref{thm:main} is trivial if $\phi_k(G) \geq \epsilon$, and
so we assume $\phi_k(G) < \epsilon$.
We also assume $\epsilon \leq 0.01$, as otherwise we reset $\epsilon=0.01$ and lose only a constant factor.


The algorithm is simple.
Set $T = \epsilon k^2\ln k/4$.
For each vertex $u$, set $p_0 = \chi_u$ and compute $S_{t, i}$ for all $t \le T$ and $i \in [n]$.
Denote these sets by $S_{t, i, u}$ to specify the starting vertex $u$.
Output a set $S = S_{t, i, u}$ that achieves the minimum in $\min_{\vol(S_{t, i, u}) \le k^{1 + \epsilon}} \phi(S_{t, i, u})$.
Clearly the algorithm runs in polynomial time.

We claim that $\phi(S) \le 4\sqrt{\phi_k(G)/\epsilon}$.
Suppose to the contrary that the algorithm does not return such a set.
Consider $t = \frac{\epsilon\ln k}{2\phi_k(G)}$;
note that $t \leq T$ as $\phi_k(G) \geq 1/k^2$ for a simple unweighted graph.
Applying Theorem~\ref{t:ub} with $l=k^{1+\epsilon}$, for any starting vertex $u$, we have
\begin{eqnarray*}
C_{t}(k) & \le & \frac{k}{k^{1 + \epsilon}} + \sqrt k(1 - 2\frac{\phi_k(G)}\epsilon)^{t}\\
& \le & k^{-\epsilon} + \sqrt k\exp(-2\frac{\phi_k(G)}\epsilon\frac{\epsilon\ln k}{2\phi_k(G)})\\
& = & k^{-\epsilon} + \sqrt k\exp(-\ln k)\\
& = & k^{-\epsilon} + k^{-\frac{1}{2}}.
\end{eqnarray*}
On the other hand, suppose $S^*$ is a set with $\vol(S^*) \le k$ and $\phi(S^*) = \phi_k(G)$.
Then Theorem \ref{t:lb} says that there exists a starting vertex $u^* \in S^*$ such that 
\begin{eqnarray*}
p_{t}(S^*) & \ge & (1 - \frac{\phi_k(G)}{2})^{t}\\
& \ge & \exp(-\phi_k(G)t) \quad ({\rm for~} \phi_k(G) < 0.01)\\
& = & \exp(-\frac12\epsilon\ln k)\\
& = & k^{-\frac{\epsilon}{2}}\\
& > & k^{-\epsilon} + k^{-\frac{1}{2}} \quad ({\rm for~} k \ge \frac1\epsilon {\rm~and~} \epsilon \leq 0.01)
\end{eqnarray*}
This is contradicting since $C_{t}(k) \ge p_{t}(S^*)$ for that starting vertex, 
completing the proof of Theorem~\ref{thm:main}(1).


Now we obtain Theorem~\ref{thm:main}(2) as a corollary of Theorem~\ref{thm:main}(1).
Set $\epsilon' = \frac{\epsilon}{2\ln k}$.
Then $k^{1+\epsilon'} \leq (1+\epsilon)k$.
By using Theorem~\ref{thm:main}(1) with $\epsilon'$, we have Theorem~\ref{thm:main}(2).


\section{Local Graph Partitioning} \label{s:local}

To implement the algorithm locally, we use truncated random walk as in \cite{spielman-teng}.
Let $q_0 = \chi_v$.
For each $t \geq 0$, 
we define $\tilde p_t$ by setting $\tilde p_t(v) = 0$ if $q_t(v) < \epsilon d(v)$ and setting $\tilde p_t(v) = q_t(v)$ if $q_t(v) \geq \epsilon d(v)$,
and we define $q_{t+1} = \tilde p_{t}W$.
Then, we just use $\tilde p_t$ to replace $p_t$ in the algorithm in Section~\ref{s:main}.
To prove that the truncated random walk algorithm works,
we first show that $\tilde p_t$ is a good approximation of $p_t$ and can be computed locally.
Then we show that the curve defined by $\tilde p_t$ satisfies the upper bound in Theorem~\ref{t:ub}, and it almost satisfies the lower bound in Theorem~\ref{t:lb}.
Finally we combine the upper bound and the lower bound to prove Theorem~\ref{thm:local}.

\subsection{Computing Truncated Distributions}

\begin{lemma}
\label{t:apd} 
There is an algorithm that compute $\tilde p_{t}$ such that $\tilde p_{t} \le p_{t} \le \tilde p_{t}(v) + \epsilon td$ for every $0 \leq t \le T$, with time complexity $O(T/\epsilon)$, where $d$ is the degree vector.
\end{lemma}

\begin{proof}
First we prove the approximation guarantee.
By induction, we have the upper bound
\[\tilde p_{t} \le q_{t} = \tilde p_{t - 1}W \le p_{t - 1}W = p_{t}.\]
Also, by induction, we have the lower bound
\[p_{t} = p_{t - 1}W \le (\tilde p_{t - 1} + \epsilon(t - 1)d)W = q_{t} + \epsilon(t - 1)d \le \tilde p_{t} + \epsilon td.\]
Next we bound the computation time.
Let $S_{t}$ be the support of $\tilde p_{t}$.
In order to compute $q_{t + 1}$ from $\tilde p_{t}$, 
we need to update each vertex $v \in S_{t}$ and its neighbors.
Using a perfect hash function, the neighbors of a vertex $v$ can be updated in $O(d(v))$ steps, and thus $q_{t + 1}$ and $\tilde p_{t+1}$ can be computed in $O(\vol(S_{t}))$ steps.
Since each vertex $v \in S_{t}$ satisfies $\tilde p_{t} \ge \epsilon d(v)$, we have $\vol(S_{t}) = \sum_{v \in S_{t}} d(v) \le p_{t}(S_{t})/\epsilon \le 1/\epsilon$, and this completes the proof.
\end{proof}

\subsection{Approximate Upper Bound}

We use the truncated probability distributions to define the curve $\tilde C_t$.
Note that $\tilde p_t$ may not be a probability distribution and $\tilde C_{t}(2m)$ may be less than one.
And we define the level sets $\tilde S_{t, i} = \{v_1, v_2, \dots, v_i\}$ when we order the vertices such that $\tilde p_{t}(v_1)/d(v_1) \ge \tilde p_{t}(v_2)/d(v_2) \ge \dots \ge \tilde p_{t}(v_n)/d(v_n)$.
We show that $\tilde C_t$ would satisfy the same upper bound as in Theorem~\ref{t:ub}.

\begin{lemma}
\label{l:aub} 
Suppose for all $t \le T$ and $i \in [n]$, we have $\phi(\tilde S_{t, i}) \ge \phi_1$ whenever $\vol(\tilde S_{t, i}) \le l \le m$.
Then $\tilde C_{t}(x) \le f_t(x) := \frac xl + \sqrt x(1 - \frac{\phi_1^2}8)^t$ for all $x \le l$.
\end{lemma}

\begin{proof}
Let $\tilde x_i = \sum_{v \in \tilde S_{t, i}} d(v)$ be the extreme points defined by $\tilde p_{t'}$.
By the same proof as in Theorem~\ref{t:ub}.
it suffices to prove that Lemma~\ref{l:bc} still holds after replacing $p_{t}$ by $\tilde p_{t}$.
It means that we need to show if $x = \tilde x_i \le m$ is an extreme point (at time $t$), $S = \tilde S_{t, j}$ is the corresponding set of vertices and $\vol(S) \ge \phi$, then $\tilde C_t(x) \le \frac12(\tilde C_{t - 1}(x - \phi x) + \tilde C_{t - 1}(x + \phi x))$.
This is true since the curve defined by $q_t = \tilde p_{t - 1}W$ is less than $\frac12(\tilde C_{t - 1}(x - \phi x) + \tilde C_{t - 1}(x + \phi x))$ by Lemma~\ref{l:bc}, and $\tilde p_t \le q_t$.
\end{proof}

\subsection{Proof of Theorem~\ref{thm:local}}


Suppose $U$ is a subset of vertices with $\vol(U) \leq k$ and $\phi(U) \leq \varphi$, where $\frac{1}{\epsilon} \leq k \leq m$.
We would prove that given $k$ and $\varphi$ and an initial vertex $u$ in $U$ with $p_t(U) \geq \frac{1}{c}(1-\frac{\phi}{2})^t$ for a constant $c>1$,
the truncated random walk algorithm will output a set $S$ with $\vol(S) \leq O(k^{1+\epsilon})$ and $\phi(S) \leq 8\sqrt{\varphi/\epsilon}$.
The running time of the algorithm is $O(\epsilon^2k^{1 + 2\epsilon}\ln^3k/\varphi^2)$.

For concreteness we set $c=4$ in the following calculations.
Set $T = \frac{\epsilon\ln k}{2\varphi}$ and $\epsilon' = \frac{k^{-1 - \epsilon}}{20T}$.
Applying Lemma \ref{t:apd} with $T$ and $\epsilon'$,
we can compute all $\tilde p_{t}$ and thus $\tilde S_{t,i}$ for all $t \leq T$ and $i \in [5k^{1+\epsilon}]$ in $O(T\ln k/\epsilon') = O(\epsilon^2k^{1 + \epsilon}\ln^3 k/\varphi^2)$ steps (with an additional $\ln k$ factor for sorting).
By Lemma~\ref{t:apd}, the starting vertex $u$ will give $\tilde p_T(U) \geq \frac{1}{4} (1-\frac{\varphi}{2})^T - \epsilon' T \vol(U)$.
We claim that one of the set $S = S_{t, i}$ must satisfy $\vol(S) \le 5k^{1 + \epsilon}$ and $\phi(S) \le 8\sqrt{\varphi/\epsilon}$.
Otherwise, setting $\phi_1 \geq 8\sqrt{\varphi/\epsilon}$, we have
\begin{eqnarray*}
\tilde p_T(U) & \geq & \frac{1}{4} (1 - \frac{\varphi}{2})^T - \epsilon' T \vol(U)\\
& \geq & \frac{1}{4} \exp(-\varphi T) - \frac{k^{-\epsilon}}{20} \quad ({\rm for~} \phi < 0.01)\\
& = & \frac{k^{-\frac{\epsilon}{2}}}{4}  - \frac{k^{-\epsilon}}{20}\\
& > & \frac{k^{-\epsilon}}{5} + k^{-\frac12} \quad ({\rm using~} k^{-\frac{\epsilon}{2}} > k^{-\epsilon} + 4k^{-\frac12} {\rm~for~} k \geq \frac{1}{\epsilon} {\rm~and~} \epsilon \leq 0.01)\\
& \geq & \frac{k}{5k^{1+\epsilon}} + \sqrt{k}(1-\frac{\phi_1^2}{8})^T\\
& \geq & \tilde C_T(k),
\end{eqnarray*}
which is a contradiction, completing the proof of Theorem~\ref{thm:local}.




\section{Concluding Remarks}

We presented a bicriteria approximation algorithm for the small sparsest cut problem with conductance guarantee independent of $n$, but the volume of the output set is $k^{1+\epsilon}$.
We note that if one can also guarantee that the volume of the output set is at most $Mk$ for an absolute constant $M$,
then one can disprove the small set expansion conjecture, which states that for any constant $\epsilon$ there exists a constant $\delta$ such that distinguishing $\phi_{\delta m}(G) < \epsilon$ and $\phi_{\delta m}(G) > 1-\epsilon$ is NP-hard.
This can be viewed as an evidence that our analysis is almost tight, or an evidence that the small set expansion problem is not NP-hard.

More formally, suppose there is a polynomial time algorithm with the following guarantee:
given $G$ with $\phi_{k}(G)$, always output a set $S$ with $\phi(S) = f(\phi_k(G))$ and $\vol(S) = Mk$ where $f(x)$ is a function that tends to zero when $x$ tends to zero (e.g. $f(x) = x^{1/100}$) and $M$ is an absolute constant.
Then we claim that there is a (small) constant $\epsilon$ such that whenever $\phi_{k}(G) < \epsilon$ there is a polynomial time algorithm to return a set $S$ with $\phi(S) < 1-\epsilon$ and $\vol(S) \leq k$.

We assume that $G$ is a $d$-regular graph, as in~\cite{raghavendra-steurer} where the small set expansion conjecture was formulated.
Suppose there is a subset $U$ with $|U|=k$ and $\phi(U) < \epsilon$.
First we use the algorithm to obtain a set $S$ with $\phi(S) \leq f(\epsilon)$ and assume $|S| = Mk$ (instead of $|S| \leq Mk$).
Next we show that a random subset $S' \subseteq S$ of size exactly $k$ will have $\phi(S') < 1-\epsilon$ with a constant probability for a small enough $\epsilon$.
Let $E(S)$ be the set of edges with both endpoints in $S$.
Each edge in $E(S)$ has probability $2(\frac{1}{M})(1-\frac{1}{M})$ to be in $\delta(S')$.
So, the expected value of 
\[|\delta(S')| \leq |\delta(S)| + 2(\frac{1}{M})(1-\frac{1}{M})|E(S)|.\]
By construction $\vol(S') = kd$, and so the expected value of
\[\phi(S') \leq \frac{|\delta(S)|}{kd} + \frac{2(\frac{1}{M})(1-\frac{1}{M})|E(S)|}{kd}.\]
Note that $|E(S)| \leq Mkd/2$ and $|\delta(S)|/kd = M\phi(S) \leq Mf(\epsilon)$, so the expected value of 
\[\phi(S') \leq Mf(\epsilon) + 1-\frac{1}{M}.\]
For a small enough $\epsilon$ depending only on $M$, the expected value of $\phi(S') \leq 1-10\epsilon$.
Therefore, with a constant probability, we have $\phi(S') < 1 - \epsilon$.
This argument can be derandomized using standard techniques.

We show that random walk can be used to obtain nontrivial bicriteria approximation algorithms for the small sparsest cut problem. 
We do not know of an example showing that our analysis is tight.
It would be interesting to find examples showing the limitations of random walk algorithms (e.g. showing that they fail to disprove the small set expansion conjecture).

\bibliographystyle{plain}

\end{document}